\documentclass[lettersize,journal]{IEEEtran}

\usepackage{color,array,amsthm}
\usepackage{graphicx}

\usepackage{cite}
\usepackage{amsmath,amssymb,amsfonts}
\usepackage{algorithmic}
\usepackage{graphicx}
\usepackage{textcomp}

\usepackage{tabularx}
\usepackage{graphicx}
\usepackage{bm}       
\usepackage{cite} 
\usepackage{float}
\usepackage{makecell}
\usepackage{xurl}
\usepackage{placeins}
\usepackage{dblfloatfix}
\newlength{\FigWidth}
\newcommand{\ignore}[1]{}  
\usepackage{comment}
\usepackage{booktabs}
\usepackage{multirow}
\usepackage{amsthm}
\newtheorem{theorem}{Theorem}
\newtheorem{lemma}[theorem]{Lemma}
\newtheorem{remark}{Remark}

\newtheorem{definition}{Definition}
\newtheorem{corollary}{Corollary}
\newtheorem{assumption}{Assumption}
\newtheorem{example}{Example}
\usepackage{algorithm}
\usepackage{algorithmic}
\usepackage{capt-of} 
\usepackage[caption=false,font=footnotesize]{subfig}
\usepackage{booktabs}
\renewcommand{\thetable}{\Roman{table}}

\begin{document}
\title{Power-Efficiency and Scalability Analysis of Magnetically-Actuated Satellite Swarms via Convex Optimization
}
\author{%
Yuta Takahashi$^{1}$, 
Seang Shim$^{2}$, 
Hiraku Sakamoto$^{1}$,
Shin-Ichiro Sakai$^{3}$
\thanks{$^{1}$ Mechanical Engineering, Institute of Science Tokyo, Meguro-ku, Tokyo 152-8550, Japan}%
\thanks{$^{2}$ Department of Space and Astronautical Science, The Graduate University for Advanced Studies, Sagamihara, Kanagawa 252-5210, Japan}
\thanks{$^{3}$Spacecraft Engineering, Institute of Space and Astronautical Science, Sagamihara, Kanagawa 252-5210, Japan}%
\thanks{Corr. author: Yuta Takahashi, {\tt\small stateofyuta@gmail.com}}%
}
\maketitle
\begin{abstract}
This correspondence presents a convex-optimization-based evaluation framework of satellite-swarm-based apertures maintained by magnetic-field interactions. Spaceborne distributed apertures are composed of multiple satellites and are attractive for scientific and commercial missions because their scalability enables high-gain, narrow-beam, and large-aperture capabilities beyond the launch-size limitations. A key challenge is that the long-term maintenance of such virtual structures requires consistent formation control amid unstable orbital dynamics, and magnetic interactions generated by satellite-mounted magnetorquers offer a desirable propellant-free position-control strategy. However, the nonlinearities of the electromagnetic force and torque model lead to a nonconvex power-consumption constraint, making system-level configuration analysis difficult. To address this issue, we develop a convex optimization-based framework to analyze the power consumption of large magnetically actuated satellite swarms. The resulting analysis shows that increasing the number of satellites can improve formation-keeping power efficiency. This indicates that magnetically actuated swarm architectures provide a power-efficient alternative to the conventional few-satellite electromagnetic formation-flight concept for constructing large-scale space systems.
\end{abstract}
\section{INTRODUCTION}
Spaceborne distributed apertures are a promising architecture for future large-scale space systems. By synthesizing an aperture from multiple satellites, they can provide high-gain, narrow-beam, and large-aperture capabilities beyond the launch-size limitations of monolithic satellites. Such capabilities are important for scientific and commercial missions, including deep-space communication, direct-connectivity services with small ground or user terminals in Fig~\ref{PSA_BW3}, and resilient communication in cellular frequency bands \cite{tuzi2023satellite,shim2025feasibility}. The membrane structure system also offers a large surface area, small volumes, and reliable deployment 
along with software-based alignment calibration \cite{you2021ka} in Fig~\ref{PSA_HELIOSR}. While the scalability of membrane structures depends on material advances, distributed apertures rely on state-estimation sensors that advance more rapidly than materials. In conventional monolithic architectures, increasing antenna size directly increases structural mass, deployment complexity, and payload volume requirements. Distributed architectures mitigate these limitations by assigning the aperture function to multiple spacecraft, while distributing payload, thermal, and structural requirements among multiple satellites \cite{hadaegh2014development}. 

One of the main difficulties in realizing such virtual space structures is the long-term maintenance of the formation, and the lifetime of a distributed aperture is limited by fuel consumption. In low Earth orbit, the relative motion is continuously affected by unstable orbital dynamics and environmental perturbations, including the $J_2$ effect \cite{schweighart2002high}. The communication performance of a distributed aperture depends directly on the accuracy of relative positions, and deviations from the desired formation can degrade the sidelobe level, effective aperture, and overall link performance. Therefore, a distributed aperture requires a formation-keeping strategy that can preserve the desired geometry over long mission durations.
\begin{figure}[!tb]
    \begin{minipage}[b]{0.56\FigWidth}
        \centering
        \subfloat[Space Rigid-panel array.]{\includegraphics[width=\linewidth]{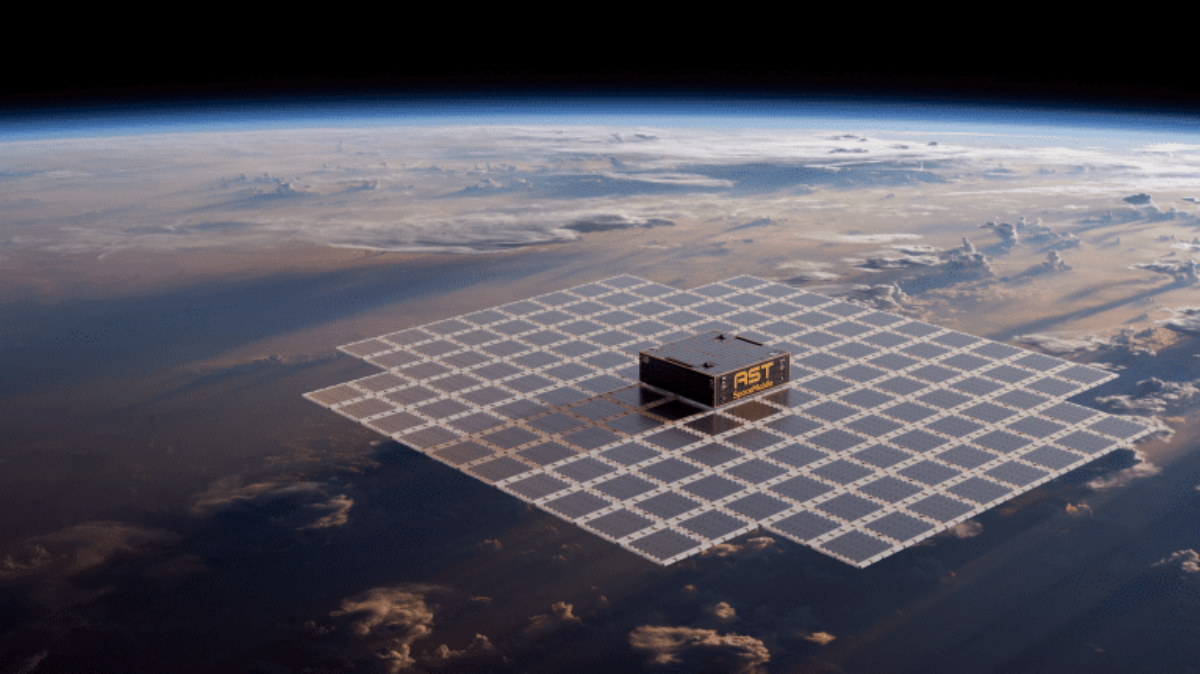}\label{PSA_BW3}}
    \end{minipage}
    \begin{minipage}[b]{0.42\FigWidth}
        \centering
        \subfloat[Space membrane array \cite{komaba2026onorbit}.]{\includegraphics[width=\linewidth]{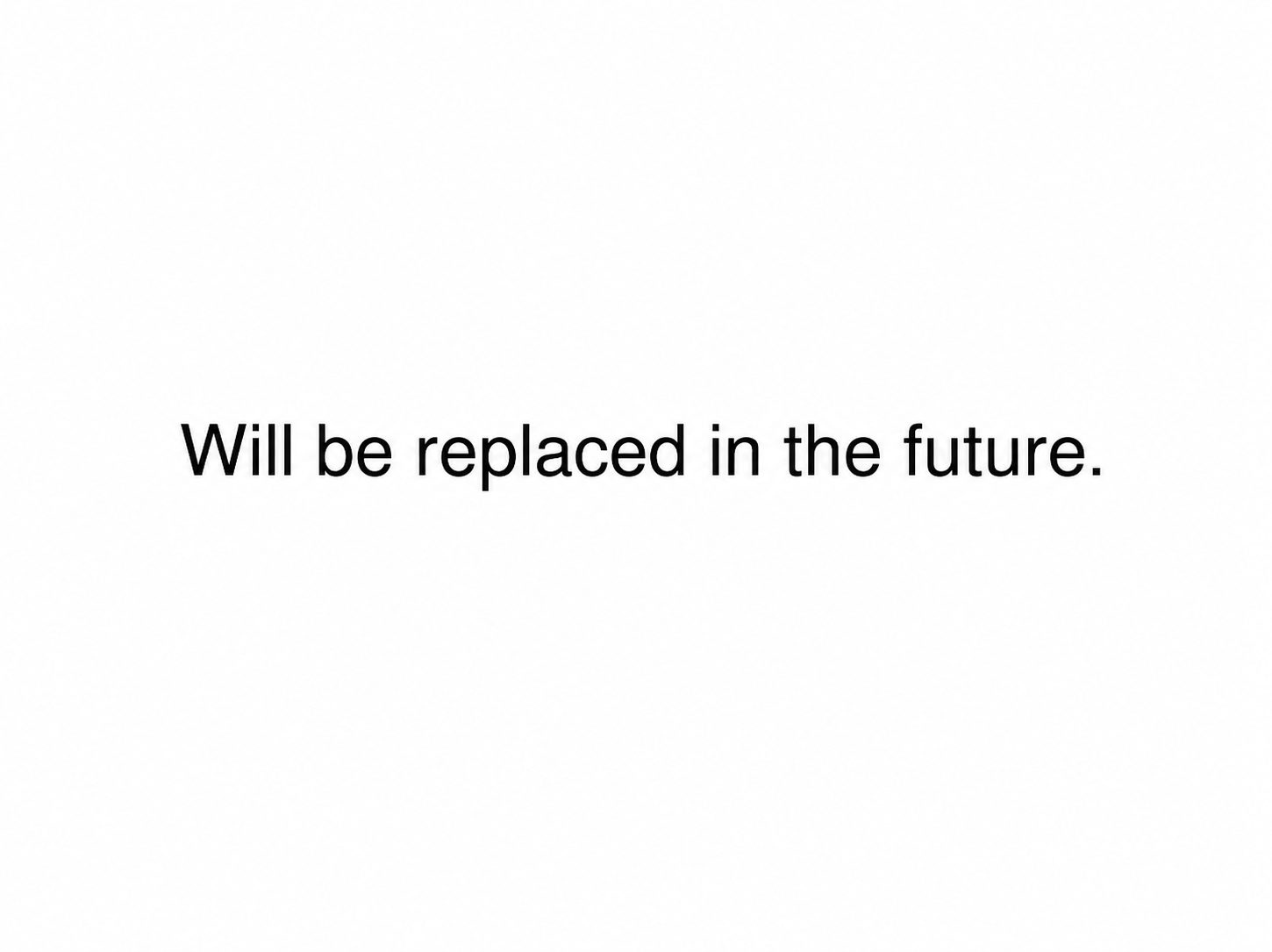}
        \label{PSA_HELIOSR}}
    \end{minipage}
    \caption{Monolithic space antenna arrays (Conceptual illustration of the BlueWalker3 © AST SpaceMobile and HELIOS-R \cite{komaba2026onorbit}).}
\end{figure}
\begin{figure}[!tb]
    \centering
{\includegraphics[width=\linewidth]{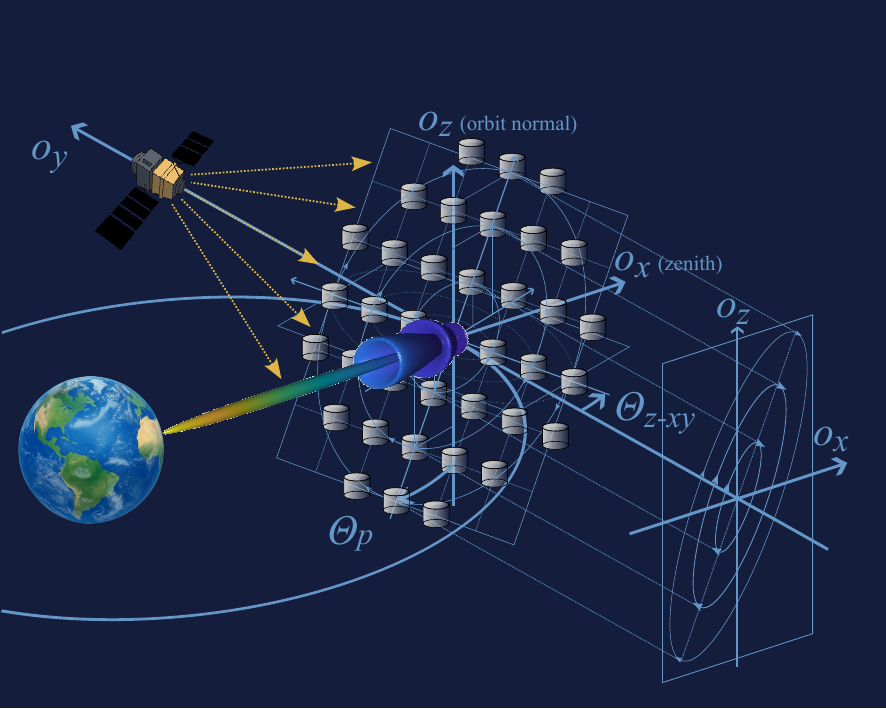}\label{PSA_homogeneous_formation_ver2}}
\vspace{-0.5cm}
    \caption{Our conceptual illustration of spaceborne distributed apertures example using satellite swarms.}
\end{figure}

A propellant-free actuation strategy is desirable for this purpose, and magnetic interactions generated by satellite-mounted magnetorquers (MTQs) provide a promising solution. MTQs are widely used as attitude actuators for Earth-orbiting satellites, and electromagnetic formation flight (EMFF) extends this magnetic actuation principle to relative-position control among multiple satellites \cite{kong2004electromagnetic,takahashi2022kinematics}. Microgravity demonstrations have validated the basic feasibility of magnetic interaction swarm control \cite{takahashi2024neural,takahashi2026certified}. Recent studies have also examined MTQ-based formation control for distributed space antenna concepts under unstable orbital dynamics \cite{takahashi2025distance,shim2025feasibility}. 

The central challenge in power evaluation for magnetically actuated swarm-keeping stems from the nonconvex nature of the magnetic interaction model. The generated force and torque are strongly coupled through the relative geometry and the dipole commands of multiple satellites. The resulting input map is bilinear, and the power required to realize a prescribed force--torque command cannot be directly inferred from independent force and torque limits. Moreover, the achievable magnetic input depends on the coil geometry and the nominal inter-satellite distance. Therefore, the formation-keeping problem in a user-defined orbit must be evaluated together with the system design variables, rather than treated as a separate control problem. This difficulty becomes more pronounced as the number of satellites increases, because the magnetic interaction network grows rapidly, and each satellite can be affected by multiple neighboring interactions. Consequently, estimating the peak and total formation-keeping power via direct nonconvex numerical optimization is computationally challenging and obscures the architecture's scalability trend. 

This correspondence addresses the issue by developing a convex-optimization-based framework to evaluate the power consumption of large magnetically actuated satellite swarms. First, we show the global optimality of the power allocation problem for two-satellite formation keeping, even though the original dipole allocation problem is nonconvex. 
Second, we combine this global optimal dipole allocation with a decentralized formation-keeping architecture. Consequently, the proposed method transforms a nonconvex system-level evaluation problem into a tractable convex-optimization-based analysis framework. 
The remainder of this correspondence is organized as follows. Section~\ref{PSA_section2} summarizes the orbital motion model and magnetic interaction model used in the analysis. Section~\ref{PSA_section3} formulates the grid-structured distributed aperture and the disturbance model for user-defined stable relative orbits. Then, we prove the global optimality of dipole allocation of the two satellite systems. Section~\ref{PSA_stable_orbit_residual_disturbances} derives the decentralized bucket-brigade model and presents the convex-optimization-based power-consumption analysis to show that increasing the number of satellites can improve formation-keeping power efficiency under the considered architecture. 
Finally, Section~\ref{conclusion} concludes this correspondence.
\section{Preliminaries}
\label{PSA_section2}
\subsection{Magnetically Actuated Swarm Control Model}
\label{magnetically_actuated_swarm_control_model}
This subsection introduces the magnetic field interaction model. We define the magnetic moment $\bm{\mu}$ and the resistance of a single-axis coil $R_{\mathrm{coil}}$ as
\begin{equation}
    \label{PSA_2_1::coil dipole moment}
        \bm{\mu} =\pi N_t a_{\mathrm{coil}}^2 c_{\mathrm{coil}}\bm{n},\quad R_{\mathrm{coil}} = {2 a_{\mathrm{coil}} N_t p_c}/{r_{\mathrm{coil}}^2}
\end{equation}
where $N_t$ is the number of coil turns, $a_{\mathrm{coil}}$ is the coil radius, $c_{\mathrm{coil}}$ is the coil current, $\bm{n}$ is a vector normal to the coil plane, $p_c$ is the wire resistivity, and $r_{\mathrm{coil}}$ is the wire radius. We assume the following sinusoidal magnetic moment for each agent \cite{takahashi2024neural}:
\begin{equation}
    \label{PSA_2_1::AC coil dipole moment}
    \bm{\mu}_{j}(t)\approx\overline{\bm{\mu}}_j \sin(\omega_{j}t+\bm{\theta}_j)=\bm{s}_{j} \sin(\omega_{{j}}t)+\bm{c}_{j} \cos(\omega_{{j}}t),
\end{equation}
where 
the amplitudes of the cosine and sine components are $\bm{s}_j\in\mathbb{R}^3$ and $\bm{c}_j\in\mathbb{R}^3$, respectively, and $\bm{\theta}_j\in \mathbb{R}^3$ are phases. 
The first-order time-averaged input $\overline{u}_{j \leftarrow k}\in\mathbb{R}^6$ exerted on the $j$-th agent by the $k$-th agent is \cite{takahashi2024neural}

\begin{equation}
\label{CGL_near_field_interaction_model}
\begin{aligned}
\overline{u}_{j \leftarrow k}
&\triangleq
\begin{bmatrix}
{f}^{\mathrm{avg}}_{j\leftarrow k}\\
{\tau}^{\mathrm{avg}}_{j\leftarrow k}
\end{bmatrix}
=\int_T\frac{\mu_0}{4\pi}Q_{j\leftarrow k}
(\mu^{b}_{k}(u)\otimes \mu^{b}_{j}(u))\frac{\mathrm{d}u}{T}\\
&\approx \frac{1}{2}\frac{\mu_0}{4\pi}Q_{j\leftarrow k}\left(s^b_{k}\otimes s^b_{j}+c^b_{k}\otimes c^b_{j}\right)\ \mathrm{if}\ \omega_j=\omega_k
\end{aligned}
\end{equation}
where $\mu_0=4\pi\times 10^{-7}$ T$\cdot$m/A and a position vector from $k$-th coil to $j$-th one ${r}_{j\leftarrow k}$ yields $Q_{{j\leftarrow k}}\in \mathbb{R}^{6\times 9}$ as \cite{takahashi2024neural}
\begin{equation*}
\begin{aligned}
&Q_{{j\leftarrow k}}=(I_2\otimes C^{A/L_{j\leftarrow k}}) 
\begin{bmatrix}
\Psi_{f}\\
\Psi_{\tau}
\end{bmatrix}
(C^{L_{j\leftarrow k}/A}\otimes C^{L_{j\leftarrow k}/A}),\\
        &\left\{
        \begin{aligned}
            &\Psi_f=
            \frac{1}{{\|r_{j\leftarrow k}\|^4}}
            {\small\begin{bmatrix}
                -6&0&0&0&3&0&0&0&3\\
                0&3&0&3&0&0&0&0&0\\
                0&0&3&0&0&0&3&0&0
                \end{bmatrix}}\\
                &\Psi_\tau=
                \frac{1}{\|r_{j\leftarrow k}\|^3}
                {\small\begin{bmatrix}
                0&0&0&0&0&1&0&-1&0\\
                0&0&2&0&0&0&1&0&0\\
                0&-2&0&-1&0&0&0&0&0
            \end{bmatrix}}
        \end{aligned}
        \right..
\end{aligned}
\end{equation*}
The line-of-sight frame $\{\mathcal{LOS}_{j\leftarrow k}\}$ \cite{takahashi2024neural} is attached to the $k$-th agent and oriented toward the $j$-th agent and  the associated coordinate transformation matrix $C^{O/L{j\leftarrow k}}\in\mathbb{R}^{3\times 3}$ is $C^{A/L_{j\leftarrow k}}=\mathcal{C}(r_{j\leftarrow k}^a,\ f^a_{j}\times r^a_{j\leftarrow k})$ where 
    \begin{equation}
        \label{PSA_new_frame}
        \mathcal{C}{(v{^a},w{^a})}=
              \begin{bmatrix}
                \mathsf{e}_x\triangleq\frac{v{^a}}{\|v{^a}\|},\mathsf{e}_y\triangleq
                \frac{{v{^a}}\times w{^a}}{\|{v{^a}}\times w{^a}\|},\mathsf{e}_x\times \mathsf{e}_y
             \end{bmatrix}.
    \end{equation}
\subsection{Passively Stable Orbits for Satellite Swarm}
\label{Passively_Stable_Orbit}
We introduce the passively stable trajectories under $J_2$ Earth gravity. 
Let ${r_{jk}}={r}_j-{r}_k=[x;y;z]$ be the relative position from the $k$-th satellite to the $j$-th one. The dynamics of the linearized relative motion in the local vertical, local horizontal (LVLH) frame are \cite{takahashi2025distance,takahashi2025scalable}
\begin{equation}
\label{PSA_Hill_dynamics}
\begin{aligned}
&\ddot{\overline{x}}-2\omega_{xy}\dot{\overline{y}}-3\omega_{xy}^2 \overline{x}-\frac{4 \omega_{xy}^2}{c_-^2/s_{J_2}}\left(2 \overline{x}+\frac{\dot{ \overline{y}}}{ \omega_{xy}}\right)=c_+(u_x+d_x) \\
& \ddot{\overline{y}}+2\omega_{xy}\dot{\overline{x}}={c_-}(u_y+d_y)\\
&\ddot{z}+\omega_z^2  z=2 l \omega_z \cos (\omega_z t+\theta_z)+(u_z+d_z) \\
&\left\{
\begin{aligned}
&\overline{x} =c_{+} x,\quad \overline{y} = c_- y,\quad{\omega}_{xy}=c_-\omega_{\mathrm{o}},\\
&\omega_z =\omega_{z\mathrm{ref}}+f_1(\delta \dot{\Omega}_{\mathrm{avg}})\approx\omega_{z\mathrm{ref}},\ r_z \sin \theta_z=z, \\
&l \sin \theta_z+\omega_z r_z \cos \theta_z=\dot{z},
\end{aligned}
\right.
\end{aligned}
\end{equation}
where $l(\delta \dot{\Omega}_{\mathrm{avg}})= -r_{\text{ref}} \sin i_j \sin i_k f_2(\delta \dot{\Omega}_{\mathrm{avg}jk}) \approx 0$ since we can naturally assume that the satellites have identical $i$. 
The analytical solution of (\ref{PSA_Hill_dynamics}) is 
\cite{takahashi2025distance,takahashi2025scalable}:
\begin{equation}
\label{PSA_CWsol}
\begin{aligned}
&\begin{bmatrix}
    {x}(t)\\
    {y}(t)\\
    {z}(t)
\end{bmatrix}
=
\begin{bmatrix}
r_{\mathrm{o}}\\
0
\end{bmatrix}
+
\begin{bmatrix}
   r_{xy}\sin{(\omega_{xy} t + \theta_{xy})}/c_+\\
    2r_{xy}\cos{(\omega_{xy} t + \theta_{xy})}/c_-\\
    (r_z+l t) \sin{({\omega}_z t+\theta_z)}
\end{bmatrix},\quad \\
&r_{\mathrm{o}}=
\begin{bmatrix}
    {2C_{1}}&{C_4-\epsilon_{2} C_1 t}
\end{bmatrix}^\top,\ \epsilon_2 = \frac{3+5s_{J_2}}{c_+c_-}\omega_{xy},
\end{aligned}
\end{equation}
where $r_{\mathrm{o}}(0,t)\in\mathbb{R}^2$ is the center position of the relative orbit. The orbital indices calculated at $t=0$ are \cite{takahashi2025distance,takahashi2025scalable}
\begin{equation}
    \label{PSA_definition_C}
    \left\{
    \begin{aligned}
    &C_{1}={c_+}/{c_-^2}(2 \overline{x}+{\dot{\overline{y}}}/{ {\omega}_{xy}}),\ C_{4}= (\overline{y}-{2 \dot{ \overline{x}}}/{{\omega}_{xy}})/c_-\\
    &r_{xy}^2= {C_2^2+C_3^2},\ \theta_{xy}= \tan^{-1}(C_3,C_2)\\ 
    &C_{2}=( \overline{y}-c_-C_4)/2, \ C_{3}=\overline{x}-2c_+C_{1}\\
    &r_{z}^2= {C_6^2+C_5^2},\ \theta_{z}= \tan^{-1}(C_6,C_5)\\
    &C_5 = \dot{z}/\omega_z,\ C_6 = z.
    \end{aligned}
    \right.
\end{equation}
Enforcing $\omega_{z}=\omega_{xy}$ and $\theta_{z}=\theta_{xy}+\tan^{-1}(2\tan \Theta_{z-xy})$ derives the desired stable trajectories $p_d(t)$ \cite{takahashi2025distance,takahashi2025scalable}
\begin{equation}
\label{PSA_desired_position}
{p}_d(t)=
\begin{bmatrix}
   (1/c_+)r_{xyd}\sin{(\omega_{xy} t + \theta_{xy})}\\
    (1/c_-)2r_{xyd}\cos{(\omega_{xy} t + \theta_{xy})}\\
    \frac{r_{xyd}}{\tan\Theta_P}\frac{\cos(\Theta_{z-xy})}{\cos(\theta_z - \theta_{xy})} \sin{({\omega}_{xy} t+\theta_{zd}(\Theta_{z-xy},t))}
\end{bmatrix}.
\end{equation}
Note that $\omega_{zd} = \omega_{xy}$ is realized via active control, or equivalently, the mismatch acts as the disturbance $d_{fz}=(\omega_{xy}^2-\omega_z^2) z$ on ${p}_d$ \cite{takahashi2025distance,takahashi2025scalable}:
\begin{equation}
\label{PSA_d_fz_disturbance}
\begin{aligned}
d_{fz}&=r_{zd}(\omega_{xy}^2 \sin{(\omega_{xy} t + \theta_z)}-\omega_{z}^2 \sin{(\omega_z t + \theta_z)}).
\end{aligned}
\end{equation}
\section{Problem Formulation: Perturbed Grid Aperture}
\label{PSA_section3}
This section formulates our problem by introducing a simplified model and states the computational problems in power estimation. We model our satellite swarm as a square formation with equally spaced satellites, as illustrated in Fig.~\ref{PSA_homogeneous_formation_ver2}. Each grid line includes a linear formation consisting of $2n+1$ satellites shown in Fig.~\ref{PSA_figure_grid_linear_approximation}, where the central satellite is indexed as 0, and the satellites at either edge are indexed as $n$ and $-n$, respectively. We define the vector from the $(-n)$th satellite to the $n$th satellite as
\begin{equation}
R_l(\tau)=r_l(\tau)\hat{p}(\tau),
\quad
\tau\in[0,2\pi/\omega_{xy}).
\end{equation}
The total number of satellites is given by
\begin{equation}
\label{PSA_N_all_square_assumption}
N_{\mathrm{all}} = N_l^2 = (2n+1)^2 ,
\end{equation}
where $N_l$ denotes the number of elements along the array. The total system mass $\overline{m}_{\mathrm{sys}}=N_{\mathrm{all}}m_{\mathrm{sat}}$, inter-distance, and array side length are user-defined constants:
\begin{equation}
\overline{m}_{\mathrm{sys}}=\mathrm{const.},\ d_{\mathrm{sat}}=\mathrm{const.},\ r_l=(2n+1)d_{\mathrm{sat}}.
\end{equation}

We assume that the on-orbit environmental disturbance force $f_d$ is given by
\begin{equation}
\label{PSA_disturbance_position_relative}
f_d=m_{\mathrm{sat}}\ K_{\mathrm{orb}}(t)\ p.
\end{equation}
where time-varying coefficient matrix $K_{\mathrm{orb}}(t)\in\mathbb{R}^{3\times 3}$ and relative position vector $p\in\mathbb{R}^{3}$. Note that we neglect the environmental disturbance forces and torques except (\ref{PSA_disturbance_position_relative}) because they are generally independent of the distance $p$ from the center and bounded.
\begin{example}
\label{PSA_example_1_2}
Consider the satellites on stable relative orbit ${p}_d(t)$ in (\ref{PSA_desired_position}). The averaging error disturbance under $J_2$ gravity is $f_d(t)=m_{\mathrm{sat}}\ K_{J_2}\ {(P_{\mathrm{ref}},i,\theta)}p$ where
\begin{equation*}
\begin{aligned}
f_d(t)&=m_{\mathrm{sat}}\left(\nabla^2 U_{J_2}-
\int_0^{2\pi} \nabla^2 U_{J_2}\frac{\mathrm{d}\theta}{2\pi}\right)+m_{\mathrm{sat}} d_{fz(t)}\\
&=m_{\mathrm{sat}}\left(K{(P_{\mathrm{ref}},i,\theta)}+
{\small\begin{bmatrix}
    0&0&0\\
    0&0&0\\
    0&0&-(\omega_z^2-\omega_{xy}^2)
\end{bmatrix}}
\right)p,
\end{aligned}
\end{equation*}
$d_{fz(t)}$ in (\ref{PSA_d_fz_disturbance}) and $K{(P_{\mathrm{ref}},i(t),\theta(t))}\in\mathbb{R}^{3\times 3}$ is \cite{schweighart2002high}
\begin{equation*}
K{(P_{\mathrm{ref}},i,\theta)}
=\frac{k_{J_2}}{2P_{\mathrm{ref}}^5}
{\small\begin{bmatrix}
12s^2_ic_{2\theta} & 4s^2_i s_{2 \theta} & 4s_{2 i} s_\theta \\
4s^2_i s_{2 \theta} & -7s^2_ic_{2\theta}& -{s_{2i} c_\theta}\\
4s_{2 i} s_\theta&-{s_{2i} c_\theta}&-5s^2_ic_{2\theta}
\end{bmatrix}}.
\end{equation*} 
\end{example}

To evaluate the required powers, one straightforward approach is to use numerical evaluations, but this becomes computationally burdensome as the number of satellites $N_{\mathrm{all}}$ increases. As shown in subsection~\ref{Passively_Stable_Orbit}, the relative orbital dynamics can be approximated as the linear time-invariant system $\dot{x}= A x(t) + d(t)$ with a time-varying external input $d(t)$. Our goal is to evaluate $x(T)$ for $N_{\mathrm{all}} \gg 1$ and its analytical solution under $x(0)=0$ is
$$
x(T) = \int_0^T e^{A(T - \tau)} d(\tau)\mathrm{d}\tau=\sum_{k=1}^K \left[ \int_{t_k}^{t_{k+1}} e^{A(T - \tau)} \mathrm{d}\tau \right] d_k
$$
where $d(t) = d_k$ for $t \in[t_k, t_{k+1})$. This method incurs a total computational cost of $\mathcal{O}(4N_{\mathrm{all}}^2K)$, which is not scalable. Alternatively, we can reduce the total cost using the eigendecomposition $A = V \Lambda V^{-1}$ and the fact that $e^{A(T - \tau)} = V e^{\Lambda(T - \tau)} V^{-1}$. This allows for the use of parallel computing, and the analytical solution is 
$$
x(T) = \int_0^T Ve^{\Lambda(T - \tau)} V^{-1} d(\tau) \mathrm{d}\tau.
$$
Although its total cost is improved as $\mathcal{O}(N_{\mathrm{all}}^3) + \mathcal{O}({N_{\mathrm{all}}^2K}/{P})$ with the parallelism $P$ and eigendecomposition cost $\mathcal{O}(N_{\mathrm{all}}^3)$, this still does not provide a scalable evaluation method. Therefore, a scalable and uniquely determined evaluation method is required for system-level design of large magnetically actuated satellite swarms.
\section{Magnetically-Actuated Swarm Keeping Evaluation}
\label{PSA_stable_orbit_residual_disturbances}
This section derives the framework to estimate the satellite's maximum power and the system's total power for formation keeping. These indices are the key factors in scalability for magnetically actuated swarms, constrained by power and thermal budgets, rather than by time-integrated consumption.
\begin{figure}[!tb]
    \centering
    \includegraphics[width=0.9\linewidth]{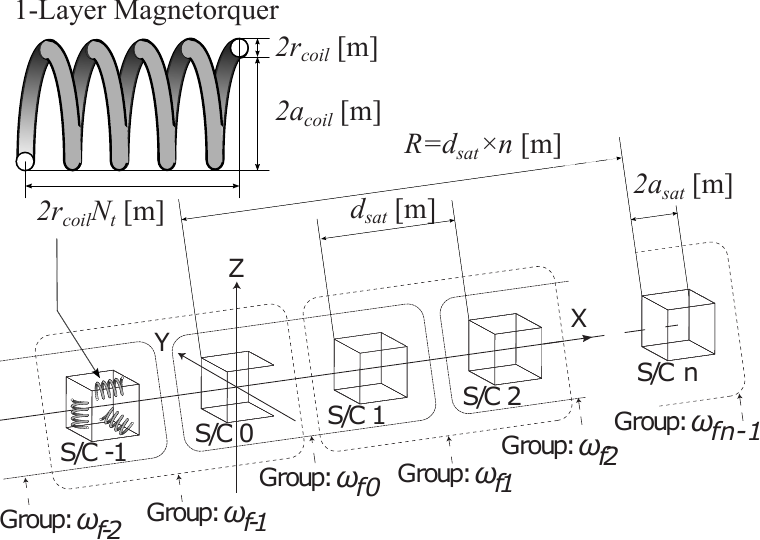}
    \caption{Grid-structured approximation for distributed space system design. Linear formation of $2n+1$ satellites with coil actuators.}
    \label{PSA_figure_grid_linear_approximation}
\end{figure}
\subsection{Optimal Decentralized Formation-Keeping Model}
We derive the decentralized control model for a magnetically-actuated swarm system. Our goal of the formation-keeping is to cancel out the disturbance force $f_{\mathrm{d}j}$ acting on each satellite by the control input $f_{\mathrm{c}j}$, i.e.,
\begin{equation}
\label{feedforward_conditions}
f_{\mathrm{d}j}=f_{\mathrm{c}j} \quad\mathrm{for}\quad j=1,\ldots,N_{\mathrm{all}}.
\end{equation}
As introduced in subsection~\ref{magnetically_actuated_swarm_control_model}, magnetic field interaction between different frequencies does not interact with each other in the first-order averaged dynamics. To represent a set of satellites driven by the same frequency, we define an arbitrary collection of satellite groups as
$$
\mathcal{G}\subseteq \{\mathfrak{g}\subseteq\{1,\ldots,N_{\mathrm{all}}\}: |\mathfrak{g}|\geq 2\}
$$
where each group $\mathfrak{g}\in\mathcal{G}$ can contain an arbitrary number of satellites. Then, the time-averaged dynamics admit nonunique feasible realizations of (\ref{feedforward_conditions}) as
\begin{equation*}
f_{\mathrm{d}j}
=
\sum_{\mathfrak{g}\in\mathcal{G}:\ j\in\mathfrak{g}}
f_{\mathrm{c}j}^{(\mathfrak{g})} \quad\mathrm{for}\quad j=1,\ldots,N_{\mathrm{all}}.
\end{equation*}
The simplest approach is to assign a single frequency to all satellites, i.e., $\mathfrak{g}=\{1,\ldots,N_{\mathrm{all}}\}$. However, finding a globally optimal solution is computationally expensive, as the optimal dipole allocation problem is an NP-hard program \cite{takahashi2024neural}, and relying on locally optimal solutions may compromise the reliability of our analysis. 

Then, we show that the globally optimal solution in the two-satellite case can be obtained by convex optimization. We consider the optimal dipole allocation for two satellites, which is a non-convex optimization with six equality constraints and belongs to the QCQP class:
\begin{equation}
\label{PSA_opt1}
\begin{aligned}
\mathrm{min}\ &J_p={\|m\|^2}/2=\|[s_j;s_k;c_j;c_k]\|^2/2\\
\mathrm{s.t.}\ &Q_{j\leftarrow k}
\left (
{s_{k}}
\otimes
{s_{j}}
+
{c_{k}}
\otimes
{c_{j}}
\right )=({8\pi}/{\mu_0})u_{j\leftarrow k}^{los}
\end{aligned}
\end{equation}
We derive its Lagrange dual problem in (\ref{PSA_opt1}) as
\begin{equation}
\label{PSA_opt2}
\mathrm{max}\ J_{d}=\frac{-\lambda^\top u^{los}_{j\leftarrow k}}{\mu_0/(8\pi)}\quad\mathrm{s.t.}\quad P_\lambda=
\begin{bmatrix}
E_3&R_\lambda\\
R_\lambda^\top &E_3
\end{bmatrix} \succeq 0
\end{equation}
where Lagrange multiplier vector $\lambda\in {\mathbb{R}}^{6}$, $R_\lambda\in\mathbb{R}^{3\times 3}$ satisfies $\mathrm{vec}(R_\lambda)=Q_{j\leftarrow k}^\top\lambda$. Despite QCQP, including our problem, not being convex, strong duality holds for QCQP with one quadratic inequality constraint provided Slater's condition holds \cite{boyd2004convex}. We show that the problem in (\ref{PSA_opt1}) holds strong duality, i.e., no duality gap exists between the primal and dual problems, and the proof is in Appendix~\ref{PSA_proof_2SC_ODA}.
\begin{lemma}
\label{PSA_lemma_2SC_ODA}
Consider the relative position $r_{j\leftarrow k}$ and the command input $u_{j\leftarrow k}^{los}=[{f_{j\leftarrow k}^{los}};{\tau_{j\leftarrow k}^{los}}]\in\mathbb{R}^6$ from $k$th agent to $j$th one. Let $\lambda^*\in\mathbb{R}^6$ be defined as the optimal Lagrange multiplier vector. Then, a global-optimum solution of the non-convex optimization is $[s_k^*,c_k^*]=-R_{\lambda^*}^{\top} [s_j^*,c_j^*]$ and $[s_j^*,c_j^*]$ satisfies
\begin{equation*}
\label{PSA_optimal_s_a_s_b}
\begin{aligned}
&{s}_j^*=\overline{\mu}_{j}\cos\theta_j,\quad {c}_j^*=\overline{\mu}_{j}\sin\theta_j\\
&\overline{\mu}_{j}=\begin{bmatrix}
\sqrt{\mathfrak{L}_1}\\\sqrt{\mathfrak{L}_2}\\\sqrt{\mathfrak{L}_3}
\end{bmatrix},
\left\{
\begin{aligned}
&|\theta_{L(1)}-\theta_{L(2)}|=\cos^{-1}({\mathfrak{L}_4}/{\sqrt{\mathfrak{L}_1\mathfrak{L}_2}})\\
&|\theta_{L(3)}-\theta_{L(1)}|=\cos^{-1}({\mathfrak{L}_5}/{\sqrt{\mathfrak{L}_3\mathfrak{L}_1}})\\
&|\theta_{L(2)}-\theta_{L(3)}|=\cos^{-1}({\mathfrak{L}_6}/{\sqrt{\mathfrak{L}_2\mathfrak{L}_3}})\\
\end{aligned}
\right.
\end{aligned}
\end{equation*}
where $\mathfrak{L}_j$ is derived by $\mathrm{vec}(\mathcal{Q}_{j\leftarrow k}^i)=Q_{j\leftarrow k}(i,:)^\top$ and
\begin{equation}
\label{PSA_R_lambda}
-\mathrm{tr}\left[R_{\lambda^*} \mathcal{Q}_{j\leftarrow k}^{i\top} 
\begin{bmatrix}
\mathfrak{L}_1&\mathfrak{L}_4&\mathfrak{L}_5\\
*&\mathfrak{L}_2&\mathfrak{L}_6\\
*&*&\mathfrak{L}_3\\
\end{bmatrix}\right]=\frac{8\pi}{\mu_0}u_{j\leftarrow k(i)}^{los},\ i\in[1,6]
\end{equation}
\end{lemma}
\begin{remark}This global optimal result is useful for nonlinear controller design for a magnetically actuated robot swarm under the strong duality assumption \cite{takahashi2025noda_mmh}.\end{remark}
To suppress unintended coupling among nonadjacent satellites in close proximity, we assume the use of multiple frequency allocations to confine electromagnetic interactions to neighboring satellites. 
\begin{assumption}
\label{PSA_AC_backet_condition}
Control pairs are defined by assigning distinct AC angular frequencies $\omega_{fk}$ for $k\in[-n,\ n]$ to adjacent satellite groups, as illustrated in Fig.~\ref{PSA_figure_grid_linear_approximation} (Please refer the detailed selection of angular frequencies \cite{takahashi2024neural}.). 
\end{assumption}
\noindent
We introduce a power index $W_{\mathrm{power}}$ \cite{takahashi2024neural} by Lemma~\ref{PSA_lemma_2SC_ODA}
\begin{equation}
\label{PSA_power_index}
W_{\mathrm{power}}\triangleq R_{\mathrm{coil}}\int_T  c_{\mathrm{coil}}^2(\tau){\mathrm{d}\tau}/{T}
=({R_{\mathrm{coil}}}/{\gamma_{\mu/c}^2})J_{\mathrm{d}}
\end{equation}
where the coil resistance $R_{\mathrm{coil}}$ in (\ref{PSA_2_1::coil dipole moment}), $\gamma_{\mu/c}$ [m$^2$] is the coil design ratio to convert $c_{\mathrm{coil}}$ into $\bm{\mu}$ in (\ref{PSA_2_1::coil dipole moment}), and the lower bound $J_d$ in (\ref{PSA_opt2}) into the $W_{\mathrm{power}}$. The definition of $\bm{\mu}$ and $R_{\mathrm{coil}}$ in (\ref{PSA_2_1::coil dipole moment}) derives by the results of Lagrange dual problem in (\ref{PSA_opt2})
$\gamma_{\mu/c}\triangleq \pi N_t a_{\mathrm{coil}}^2$ and $\tfrac{R_{\mathrm{coil}}}{\gamma_{\mu/c}^2}
=\tfrac{{2p_c}/r_{\mathrm{coil}}^2}{\pi^2  N_t a_{\mathrm{coil}}^3}$.
\subsection{Recursive Disturbance Elimination Model}
We estimate the maximum electric power required to eliminate the orbital disturbances in (\ref{PSA_disturbance_position_relative}) via magnetic-field interactions. We use a simplified analytical model introduced in the previous subsection. We derive $j$th error disturbance introduced in subsection~\ref{PSA_section3}
\begin{equation}
\label{PSA_disturbance_position_relative_2}
f_{d(j)}=m_{\mathrm{sat}}K_{\mathrm{orb}}(t)\ (jd_{\mathrm{sat}}\hat{p}),\quad \|\hat{p}\|_2=1
\end{equation}
where $d_{\mathrm{sat}}\hat{p}$ is the constant spacing. We define the equilibrium conditions between adjacent satellites.
\begin{definition}[
``Bucket-Brigade'' Model]
\label{PSA_Bucket_Brigade_model}
Equilibrium conditions of linear formation under Assumption~\ref{PSA_AC_backet_condition} are 
\begin{equation}
\label{PSA_Equilibrium_conditions}
    \begin{aligned}
        f_{(j)\leftarrow (j-1)} + f_{d(j)} + f_{(j)\leftarrow (j+1)}
        &= 0\\
        \tau_{(j)\leftarrow (j-1)} + \tau_{d(j)} + \tau_{(j)\leftarrow (j+1)}
        &= 0\\
    \end{aligned}
\end{equation}
where we use $f_{d(j)}$ in (\ref{PSA_disturbance_position_relative_2}) and $\tau_{d(j)}=0$.
\end{definition}
Then, we uniquely derive the required feedforward input for a recursive environmental disturbance elimination, and the proof is deferred to Appendix~\ref{PSA_proof_Bucket_Brigade_Model}.
\begin{lemma}
\label{PSA_lemma_Bucket_Brigade_Model}
The required magnetic force and torque for the recursive disturbance elimination in Definition~\ref{PSA_Bucket_Brigade_model} are 
\begin{equation}
\label{PSA_hat_U_time_varying}
\begin{aligned}
&u_{(j-2)\leftarrow (j-1)}
\triangleq\ \chi_{\mathrm{sys}}(\overline{m}_{\mathrm{sys}},n)L{(n,j)}\hat{U}{({r_l},t)}\\
&
\begin{aligned}
&\chi_{\mathrm{sys}}(\overline{m}_{\mathrm{sys}},n)\triangleq \overline{m}_{\mathrm{sys}}\frac{n(n+1)}{6(2n+1)^3}\in\mathbb{R}\\
&L_{(n,j)}\triangleq
    \begin{bmatrix}
\frac{(n-j+2)(n+j-1)}{n(n+1)}I_3&O_3\\
O_3&\frac{(n-j+2)(n-j+3)(2n+j-1)}{n(n+1)(2n+1)}I_3
\end{bmatrix}\\
&\hat{U}{({r_l},t)}\triangleq
\begin{bmatrix}
3K_{\mathrm{orb}}(t)R_l(t)\\
R_l(t) \times K_{\mathrm{orb}}(t)R_l(t)
\end{bmatrix}
\in\mathbb{R}^{6}
\end{aligned}
\end{aligned}
\end{equation}
where $j{\in[2,\ n+1]}$, $L{(n,\ j)}\in\mathbb{R}^{6\times 6}$, and orbital disturbance coefficient matrix $K_{\mathrm{orb}}(t)$ in (\ref{PSA_disturbance_position_relative}).
\end{lemma}
\subsection{Convex-Optimization-based Power Estimation}
We finally derive a power consumption for formation keeping. We derive the power distribution trend required by the overall distributed space systems.
\begin{theorem}
\label{PSA_averaged_total_power}
Consider the grid-structured satellites on the user-defined plane with ${\overline{m}_{\mathrm{sys}}}=(2n+1)^2m_{\mathrm{sat}}$ and ${r_l}=(2n+1)d_{\mathrm{sat}}$. Then, the upper-bound of maximum power consumption $\overline{W}$ [W] along all satellites is
\begin{equation}
\label{PSA_maximum_power_consumption}
\overline{W}\triangleq \chi_{\mathrm{sys}(\overline{m}_{\mathrm{sys}},n)}\sup_{t\in\left[0,\ T_{\mathrm{orb}}\right)}\left[w^*_{({r_l},n,2,t)}= w^*_{({r_l},t)}\right]
\end{equation}
where $w^*$ is given by a convex optimization for $j_{\in[2,n+1]}$:
\begin{equation}
\label{PSA_symmetric_liner_formation_pair}
\begin{aligned}
&w^*_{({r_l},n,j,t)}\triangleq\frac{2R_{\mathrm{coil}}}{\gamma_{\mu/c}^2}\max_{\substack{\lambda_-\in\mathbb{R}^6\\
\lambda_+\in\mathbb{R}^6}}\sum_{x=\{-,+\}}
\frac{\lambda_x^\top L_{(n,\ j)}\hat{U}_{\left({r_l},t_{x}\right)}}{-{\mu_0}/{(8\pi)}}\\
&\ \quad\ \mathrm{s.t.}\ 
\left\{
\begin{aligned}
&
\begin{bmatrix}
E_3&R_{\lambda_+}\\
R_{\lambda_+}^\top &E_3
\end{bmatrix} \succeq 0,\ 
\begin{bmatrix}
E_3&R_{\lambda_-}\\
R_{\lambda_-}^\top &E_3
\end{bmatrix} \succeq 0
\end{aligned}
\right..
\end{aligned}
\end{equation}
\end{theorem}
\begin{proof}
The power consumption between $(j-2)$th and $(j-1)$th satellites for $j{\in[2,\ n+1]}$ under the decentralized control of Assumption~\ref{PSA_AC_backet_condition} is derived the convex optimizaiton:
\begin{equation}
\label{PSA_linear_formation_power_estimation}
\max_{\lambda \in \mathbb{R}^{6}}\ J_{\mathrm{d}}=\frac{\lambda^\top u_{(j-2)\leftarrow (j-1)}(t)}{-\mu_0/(8\pi)}\ \mathrm{s.t.}\quad 
\begin{bmatrix}
E_3&R_\lambda\\
R_\lambda^\top &E_3
\end{bmatrix} \succeq 0
\end{equation}
where $r_{(j-2)\leftarrow (j-1)}(t) = - d_{\mathrm{sat}}\hat{p}(t)$ derives $R_\lambda$ 
$$
\mathrm{vec}(R_\lambda(t))=Q^\top_{{(j-2)\leftarrow (j-1)}}(- d_{\mathrm{sat}}\hat{p}(t))\lambda.
$$
Each satellite belongs to two linear formations under the decentralized control scheme and another linear formation exists as the $\pi/2$-phase-shifted linear formation. We define these two phase-shifted times $t_{-,+}\in\mathbb{R}$ in (\ref{PSA_two_phase_shifted_time}):
\begin{equation}
    \label{PSA_two_phase_shifted_time}
\forall t,\quad  t_- \triangleq t,\quad t_+ \triangleq t+{T_{\mathrm{orb}}}/{4},\quad T_{\mathrm{orb}}\triangleq{2\pi}/{\omega_{xy}}
\end{equation} 
Since the linear formation spans satellites from $-n$ to $n$, we account for the symmetric satellite pairs from $-n$ to $ 0$ by multiplying (\ref{PSA_linear_formation_power_estimation}) by a factor of two. Therefore, the total power between $(j-2)$th and $(j-1)$th satellite pair and its symmetric $(-j+2)$th and $(-j+1)$th satellite pair for $j{\in[2,\ n+1]}$ 
is calculated as $\chi_{\mathrm{sys}(\overline{m}_{\mathrm{sys}},n)}w^*_{({r_l},n,j,t)}$ [W] through the convex optimization in (\ref{PSA_symmetric_liner_formation_pair}). The coefficients in (\ref{PSA_hat_U_time_varying}) show that the maximum control force and torque is given $j=2$ and work on the center satellite $u_{(j-2)\leftarrow (j-1)}=\chi_{\mathrm{sys}}(\overline{m}_{\mathrm{sys}},n)\hat{U}$. Thus, the problem in (\ref{PSA_linear_formation_power_estimation}) derives $\overline{W}$ along all satellites in (\ref{PSA_maximum_power_consumption}).
\end{proof}
\begin{figure}[!tb]
    \centering
    \begin{minipage}[b]{1\FigWidth}
        \centering
        \subfloat[
        Averaged squared dipole moment $M(r_l,n)$ 
        for formation keeping.]{\includegraphics[width=\linewidth]{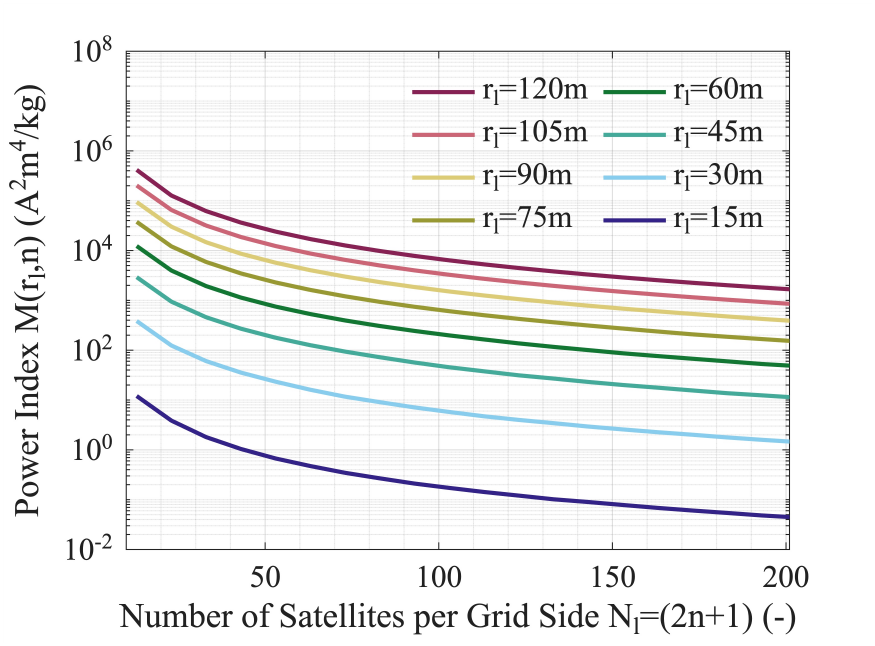}\label{PSA_averaged_power}}
    \end{minipage}\\
    \begin{minipage}[b]{1\FigWidth}
        \centering
        \subfloat[Surface area ratio of the distributed system to the monolithic system.]{\includegraphics[width=\linewidth]{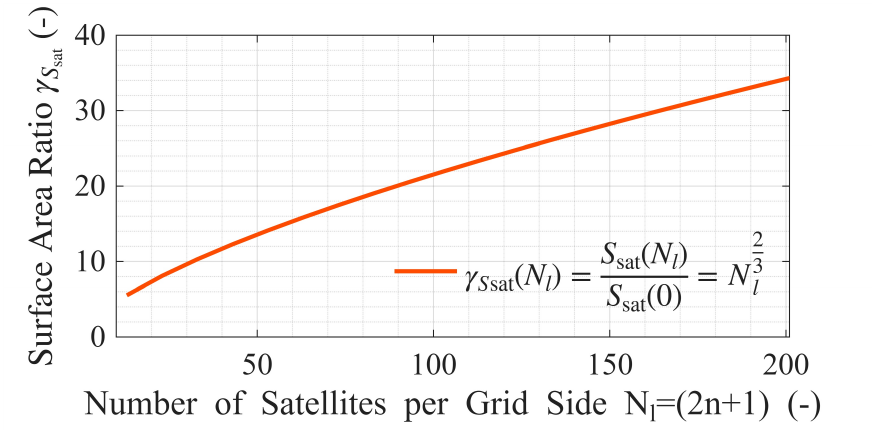}\label{PSA_surface_are_ratio}}
    \end{minipage}
    \caption{The averaged total power consumption per total system mass $m_{\mathrm{sys}}$ for formation keeping during $T_{\mathrm{orb}}$ and the ratio $\gamma_{S_{\mathrm{sat}}} (N_l)$ in (\ref{PSA_ratio_surface_area}) of total surface $S_{\mathrm{sat}}(n)$ with the monolithic structure one.}
    \label{PSA_fig: optimization results case 11}
\end{figure}
\section{Discussion}
This section presents the underlying trends and trade-off between the power consumption and the number of satellites for magnetic formation-keeping. 
\subsection{Peak Power and Number of Satellites Tradeoff}
We investigate the trend in peak power for magnetic formation keeping. The results in Theorem~\ref{PSA_averaged_total_power} show that, 
as the total system mass $\overline{m}_{\mathrm{sys}}$ increases, the upper bound of power consumption grows linearly regardless of the array size $r_l$ or the number of satellites $N_{\mathrm{all}}$, which is evident from ${W}_{\oint}$ in (\ref{PSA_averaged_power_consumption}) and $\overline{W}$ in (\ref{PSA_maximum_power_consumption}). Moreover, since $w^*_{({r_l},n,n,t)}$ is independent of $n$ if $j=n$, i.e., $w^*_{({r_l},n,n,t)}= w^*_{({r_l},t)}$, the result in (\ref{PSA_maximum_power_consumption}) indicates that $\overline{W}$ is linear with $\chi_{\mathrm{sys}(\overline{m}_{\mathrm{sys}},n)}$ along with the constant $\sup_{t\in\left[0,\ T_{\mathrm{orb}}\right)}w^*_{({r_l},t)}$ for given formation length $r_l$. Since $\chi_{\mathrm{sys}(\overline{m}_{\mathrm{sys}},n)}\rightarrow 0$ as $n\rightarrow \infty$, $\overline{W}$ converges to zero and the formation-keeping constraints imposed on the center satellite become less restrictive. This indicates the advantage of a distributed architecture based on magnetic-field interactions. 
\subsection{Minimization Trend of Total Power Consumption}
\label{total_power_consumption}
The total power budget ${W}_{\oint}$ in (\ref{PSA_averaged_power_consumption}) is evaluated numerically since an analytical expression is not available.
\begin{corollary}
Consider the grid-structured satellites on the user-defined plane with ${\overline{m}_{\mathrm{sys}}}=(2n+1)^2m_{\mathrm{sat}}$ and ${r_l}=(2n+1)d_{\mathrm{sat}}$. The upper-bound of its averaged total power consumption ${W}_{\oint}(\overline{m}_{\mathrm{sys}},{r_l},n)$ for formation keeping during one orbit $T_{\mathrm{orb}}$ is 
\begin{equation}
\label{PSA_averaged_power_consumption}
{W}_{\oint}\triangleq{\chi_{\mathrm{sys}\ (\overline{m}_{\mathrm{sys}},n)}}\int_{0}^{T_{\mathrm{orb}}}(2n+1)\sum_{j=2}^{n+1}\ {w^*_{({r_l},n,j,t)}}\ \frac{\mathrm{d}t}{T_{\mathrm{orb}}}
\end{equation}
where $w^*_{({r_l},n,j,t)}$ for $j_{\in[2,n+1]}$ is the optimal solution of the convex optimization in (\ref{PSA_symmetric_liner_formation_pair}).
\end{corollary}
\begin{proof}
The total power between $(j-2)$th and $(j-1)$th satellite pair (and its symmetric $(-j+2)$th and $(-j+1)$th satellite pair) for the orthogonal linear formation is calculated as $\chi_{\mathrm{sys}(\overline{m}_{\mathrm{sys}},n)}w^*_{({r_l},n,j,t)}$ [W] through the convex optimization in (\ref{PSA_symmetric_liner_formation_pair}). Then, (\ref{PSA_averaged_power_consumption}) sums the power distribution trend required by the overall distributed system.
\end{proof}
We define the averaged squared dipole moment $M(r_l,n)$ as the averaged total power consumption, normalized by the system mass $\overline{m}_{\mathrm{sys}}$ and scaled by $R_{\mathrm{coil}}/\gamma_{\mu/c}^2$
$$
M({r_l},n)\triangleq{{W}_{\oint}(\overline{m}_{\mathrm{sys}},{r_l},n)}/({{\overline{m}_{\mathrm{sys}}} ({R_{\mathrm{coil}}}/{\gamma_{\mu/c}^2}) })
$$
and Figure~\ref{PSA_averaged_power} shows its values for an altitude of $500$ km, an inclination of $45^\circ$, $\theta_0=0$, and a passively stable orbital plane of $(\theta_P,\theta_{ZmXY})=(30^\circ,0^\circ)$. Although the total power scales with $(2n+1)^2$, the formation-keeping power decreases as the number of satellites increases; similar trends were observed for other parameter settings. This suggests that, for magnetic-field-interaction-based formation keeping, increasing the number of satellites while holding total mass and array size fixed is energetically favorable because shorter inter-satellite distances improve magnetic actuation efficiency. 
\subsection{Maximization Trend of Solar Panel Area}
Moreover, the increased number of satellites increases the available electrical power and communication performance. For a given constant satellite volume $V_{\mathrm{sys}}$, we can derive satellite size $2a_{\mathrm{sat}}$ and total surface area of overall satellites $S_{\mathrm{sat}}\triangleq N_l^2\times 6(2a_{\mathrm{sat}})^2$ as
$$
2a_{\mathrm{sat}}\triangleq\left(\frac{V_{\mathrm{sys}}}{N_l^2}\right)^{1/3}\Rightarrow\ S_{\mathrm{sat}}(N_l)=
\frac{6N_l^{2}}{N_l^{\frac{4}{3}}}V_{\mathrm{sys}}^{\frac{2}{3}}=
{6N_l^{\frac{2}{3}}}V_{\mathrm{sys}}^{\frac{2}{3}}.
$$
For $n$, the ratio of total surface $S_{\mathrm{sat}}(n)$ with the monolithic space structure one, i.e., $S_{\mathrm{sat}}(0)=6V_{\mathrm{sys}}^{2/3}$, is 
\begin{equation}
\label{PSA_ratio_surface_area}
\gamma_{S_{\mathrm{sat}}} (N_l)\triangleq {S_{\mathrm{sat}}(N_l)}/{S_{\mathrm{sat}}(0)}= {N_l^{\frac{2}{3}}}
\end{equation}
and this is illustrated in Fig.~\ref{PSA_surface_are_ratio}. This increase in the total surface area improves both electrical power and radio-frequency performance, including antenna gain, sidelobe level, and effective isotropic radiated power.
\subsection{Limitations and Future Work}
Excessively increasing the number of satellites would eventually lead to unrealistically small satellite sizes in practice. We can formulate the detailed design problem as a potentially nonconvex optimization \cite{shim2025feasibility}, yielding feasible design solutions and more refined scaling trends. 
Moreover, environmental disturbances cannot be perfectly modeled; nevertheless, stabilizing the satellites within a tolerable position error reduces power consumption, rather than relying on ideal feedforward.

The worst-case power estimate in (\ref{PSA_hat_U_time_varying}) is not directly suitable for satellite design. 
From the disturbance model in (\ref{PSA_disturbance_position_relative_2}), the largest disturbance within a formation is
$$
\sup_j f_{d(j)}=f_{d n}=\frac{{m}_{\mathrm{sat}}}{2}K_{\mathrm{orb}}(t)R_l(t)
$$
at the satellite located farthest from the formation center. In contrast, the largest feedforward force in (\ref{PSA_hat_U_time_varying}) increases approximately in proportion to $n$
$$
\sup_j f_{(j-2)\leftarrow (j-1)}=f_{0 \leftarrow 1}
=\frac{n(n+1)}{(2n+1)}\sup_j f_{d(j)}
$$
Note that the total system power decreases as the number of satellites increases, as shown in subsection~\ref{total_power_consumption}, and such an excessive force is canceled out among different groups, i.e., $f_{0 \leftarrow 1}+f_{0 \leftarrow -1}=0$. However, each group must independently generate such large intermediate forces in our decentralized control model. For a homogeneous satellite swarm, the design should be based on the worst-case satellite, leading to overly conservative component specifications. For detailed satellite configuration design, a different evaluation metric is required that better reflects the actual per-satellite power distribution.
\section{CONCLUSION}
\label{conclusion}
This correspondence presented a convex-optimization-based framework for evaluating the formation-keeping power of magnetically actuated distributed apertures. We showed that the nonconvex two-satellite dipole-allocation problem admits a uniquely characterizable global optimum through a convex formulation, and incorporated this result into a decentralized bucket-brigade model for large satellite swarms. The resulting analysis indicated that increasing the number of satellites can improve formation-keeping power efficiency under the considered grid-structured architecture. These results support the use of magnetically actuated swarms as a power-efficient approach for constructing large-scale space systems. 
\section*{APPENDIX}
\subsection{Proof of Lemma~\ref{PSA_lemma_2SC_ODA}}
\label{PSA_proof_2SC_ODA}
First, we show that the non-convex primal problem and its associated Lagrange dual problem have zero duality gap for the considered feasible command input. Equivalently, there exists an optimal Lagrange multiplier $\lambda^*$ such that the corresponding Lagrangian provides a tight global lower bound for the primal objective. Using Roth's column lemma, $\lambda^\top Q_{j\leftarrow k}({s_k}\otimes{s_j}+{c_k}\otimes{c_j})$ is converted into $\mathrm{tr}[R_\lambda^\top(s_js_k^\top+c_jc_k^\top)]$ where $R_\lambda$ satisfies $\mathrm{vec}(R_\lambda)=Q_{j\leftarrow k}^\top\lambda$. Thus, the Lagrangian is written as $L(m,\lambda)=\frac{1}{2}m^\top (I_2\otimes P_\lambda)m-\frac{8\pi}{\mu_0}\lambda^\top u^{los}_{j\leftarrow k}$ where $P_\lambda=[I_3, R_\lambda; R_\lambda^\top; I_3]$. For the dual optimal multiplier $\lambda^*$, we have $P_{\lambda^*}\succeq 0$. Therefore, $L(m,\lambda^*)$ gives a global lower bound of the primal objective. Hence, any feasible point satisfying the stationarity condition with $\lambda^*$ attains this lower bound and is globally optimal. The first-order optimality conditions are ${\partial L}/{\partial \lambda}=g(m)=0$ and ${\partial L}/{\partial m}=(I_2\otimes P_\lambda)m=0$.  The second equality yields $P_\lambda
[s_j;s_k]=P_\lambda[c_j;c_k]=0$. Thus,
$$
[s_j,c_j]+R_\lambda [s_k,c_k]=0,\ [s_k,c_k]+R_\lambda^\top [s_j,c_j]=0.
$$
In particular, at the optimal multiplier $\lambda^*$, $[s_k^*,c_k^*]=-R_{\lambda^*}^{\top}[s_j^*,c_j^*]$. This relationship also yields
$$
[s_j^*,c_j^*]
=
R_{\lambda^*}R_{\lambda^*}^{\top}[s_j^*,c_j^*],\ [s_k^*,c_k^*]
=
R_{\lambda^*}^{\top}R_{\lambda^*}[s_k^*,c_k^*],
$$
and therefore
$\|s_j^*\|= \|R_{\lambda^*}s_k^*\|=\|s_k^*\|,\ \|c_j^*\|= \|R_{\lambda^*}c_k^*\|=\|c_k^*\|$. Consequently, global-optimum solutions satisfy $\|s_k^*\|^2+\|c_k^*\|^2=\|s_j^*\|^2+\|c_j^*\|^2=\mathrm{const}$.

Next, introduce the lifted variables $\mathfrak{L}_i$ by
$$
G_j^*\triangleq s_j^* s_j^{*\top}+c_j^*c_j^{*\top}
\triangleq
\begin{bmatrix}
\mathfrak{L}_1&\mathfrak{L}_4&\mathfrak{L}_5\\
*&\mathfrak{L}_2&\mathfrak{L}_6\\
*&*&\mathfrak{L}_3
\end{bmatrix}.
$$
where satisfies $G_j^*\succeq 0$ and $\mathrm{rank}(G_j^*)\leq 2$. For arbitrary $x\in\mathbb{R}^9$ and $X\in\mathbb{R}^{3\times 3}$ satisfying $\mathrm{vec}(X)=x$, Roth's column lemma gives
$$
\begin{aligned}
x^\top(s_k^*\otimes s_j^*+c_k^*\otimes c_j^*)
&=
\mathrm{tr}\left[
X^\top(s_j^*s_k^{*\top}+c_j^*c_k^{*\top})
\right] \\
&=
-\mathrm{tr}\left[
R_{\lambda^*}X^\top
(s_j^*s_j^{*\top}+c_j^*c_j^{*\top})
\right] \\
&=
-\mathrm{tr}\left[
R_{\lambda^*}X^\top G_j^*
\right].
\end{aligned}
$$
Replacing $x$ in the above relationship by $Q_{j\leftarrow k}(i,:)^\top\in\mathbb{R}^9$ 
reduces the primal equality constraints to (\ref{PSA_R_lambda}). Finally, because $G_j^*=s_j^*s_j^{*\top}+c_j^*c_j^{*\top}$, there exist amplitudes and phases such that $s_j^*=\overline{\mu}_{j}\cos\theta_j$ and $c_j^*=\overline{\mu}_{j}\sin\theta_j$ in Lemma~\ref{PSA_lemma_2SC_ODA}. Together with $[s_k^*,c_k^*]=-R_{\lambda^*}^{\top}[s_j^*,c_j^*]$, this gives the global-optimum dipole allocation.
\subsection{Proof of Lemma~\ref{PSA_lemma_Bucket_Brigade_Model}}
\label{PSA_proof_Bucket_Brigade_Model}
The boundary conditions restrict the edge $(n)$th and $(-n)$th satellites to interact only with their neighboring $(n-1)$th and $(-n+1)$th satellites, respectively. These derive the required electromagnetic force and torque:
$$
\begin{aligned}
f_{(n)\leftarrow(n-1)} + f_{d(n)}&\triangleq 0,\quad f_{(-n)\leftarrow (-n+1)}+f_{d(-n)} \triangleq 0\\
\tau_{(n)\leftarrow(n-1)} &= \tau_{(-n)\leftarrow (-n+1)} \triangleq 0
\end{aligned}
$$
The angular momentum conservation satisfies the interaction between neighboring satellites:
\begin{equation*}
\tau_{(j)\leftarrow(j-1)}+\tau_{(j-1)\leftarrow(j)}+\frac{d_{\mathrm{sat}}}{2}\hat{p}\times(f_{(j)\leftarrow(j-1)}-f_{(j-1)\leftarrow(j)})=0
\end{equation*}
and the summary of the linear and angular momentum conservation between $j$ and $k$th neighboring satellites is
\begin{equation}
\label{PSA_linear_and_angular_momentum_conservation}
\left\{
\begin{aligned}
&f_{(j)\leftarrow (j-1)}+f_{(j-1)\leftarrow(j)}=0\\
&\tau_{(j)\leftarrow(j-1)}+\tau_{(j-1)\leftarrow(j)}+{d_{\mathrm{sat}}}\ \hat{p}\times f_{(j)\leftarrow(j-1)}=0\\
\end{aligned}
\right.
\end{equation}
where we use $f_{(k)\leftarrow(j)}=-f_{(j)\leftarrow (k)}$ for angular momentum conservation. The conservation in (\ref{PSA_linear_and_angular_momentum_conservation}) derive reaction force and torque in the $(n-1)$th satellites: 
$$
\begin{aligned}
f_{(n-1)\leftarrow(n)} &= -f_{(n)\leftarrow(n-1)}=f_{d(n)}\\
\tau_{(n-1)\leftarrow(n)}&=-\tau_{(n)\leftarrow(n-1)}-{d_{\mathrm{sat}}}\ \hat{p}\times f_{(n)\leftarrow(n-1)}\\
&=-{d_{\mathrm{sat}}}\ \hat{p}\times f_{(n)\leftarrow(n-1)}
\end{aligned}
$$
The conservation in (\ref{PSA_linear_and_angular_momentum_conservation}) also derive reaction force and torque in the $(n-2)$th satellites as follows
$$
\begin{aligned}
f_{(n-2)\leftarrow(n-1)} =& -f_{(n-1)\leftarrow(n-2)}=f_{d(n-1)}+f_{(n-1)\leftarrow(n)}\\
\tau_{(n-2)\leftarrow(n-1)}=&-\tau_{(n-1)\leftarrow(n-2)}-{d_{\mathrm{sat}}}\ \hat{p}\times f_{(n-1)\leftarrow(n-2)}\\
=&{d_{\mathrm{sat}}}\ \hat{p}\times (f_{(n-1)\leftarrow(n)}+f_{(n-2)\leftarrow(n-1)})\\
\end{aligned}
$$
where we use equilibrium condition $\tau_{(n-1)\leftarrow (n-2)}+\tau_{(n-1)\leftarrow (n)}=0$. Generalizing these for the $j_{\in[2,n+1]}$th satellite derives the required electromagnetic force as
\begin{equation}
    \label{PSA_general_force_balance}
    \begin{aligned}
        f_{(j-2)\leftarrow (j-1)}&=-f_{(j-1)\leftarrow (j-2)}= f_{d(j-1)} + f_{(j-1)\leftarrow (j)}\\
        &= \sum_{k=j-1}^nf_{d(k)}= \sum_{k=j-1}^n m_{\mathrm{sat}}K_{\mathrm{orb}}(t)(kd_{\mathrm{sat}}\hat{p})\\
        &=\frac{(n-j+2)(n+j-1)}{2/(m_{\mathrm{sat}}d_{\mathrm{sat}})} K_{\mathrm{orb}}(t)\hat{p}
    \end{aligned}
\end{equation}    
and the required electromagnetic torque as
\begin{equation}
\label{PSA_general_torque_balance}
      \begin{aligned}
       &\tau_{(j-2)\leftarrow (j-1)}
       ={d_{\mathrm{sat}}}\ \hat{p}\times \sum_{k=j-1}^{n} f_{(k-1)\leftarrow(k)}\\
=&{d_{\mathrm{sat}}}\ \hat{p}\times \sum_{k=j-1}^{n}\frac{(n-k+1)(n+k)}{2}m_{\mathrm{sat}}d_{\mathrm{sat}} K_{\mathrm{orb}}(t)\hat{p}\\
=&\frac{(n-j+2)(n-j+3)(2n+j-1)}{6/(m_{\mathrm{sat}}d_{\mathrm{sat}}^2)}\ \hat{p}\times K_{\mathrm{orb}}(t)\hat{p}
    \end{aligned}
\end{equation}
Following the bucket-brigade logic, the disturbances from both sides of the linear array accumulate and are ultimately canceled at the central $0$th satellite, i.e.,
$$
f_{(0)\leftarrow (1)}+f_{(0)\leftarrow (-1)}=0,\quad\tau_{(0)\leftarrow (1)}+\tau_{(0)\leftarrow (-1)}=0
$$
where $f_{d(0)}=0$. Applying ${\overline{m}_{\mathrm{sys}}}=(2n+1)^2m_{\mathrm{sat}}$ and ${r_l}=(2n+1)d_{\mathrm{sat}}$ into (\ref{PSA_general_force_balance}) and (\ref{PSA_general_torque_balance}) yields (\ref{PSA_hat_U_time_varying}).
\bibliographystyle{IEEEtaes} 
\bibliography{references} 
\end{document}